\pdfoutput=1
\documentclass[sigconf,screen]{acmart}

\usepackage[normalem]{ulem}

\newcommand{\bolda}{{\boldsymbol{a}}}
\newcommand{\boldb}{{\boldsymbol{b}}}
\newcommand{\boldc}{{\boldsymbol{c}}}

\newcommand{\boldx}{{\boldsymbol{x}}}

\usepackage{bbm}
\usepackage[ruled, algo2e, vlined]{algorithm2e}

\usepackage{tcolorbox}

\SetCommentSty{mycommfont}

\newtheorem{theorem}{Theorem}

\newtheorem{proposition}[theorem]{Proposition}

\theoremstyle{definition}

\usepackage{multirow}
\usepackage{tabularx}

\usepackage{pifont}
\newcommand{\cmark}{{\color[HTML]{03af7a} \ding{51}}}%
\newcommand{\xmark}{{\color[HTML]{ff4b00} \ding{55}}}%

\allowdisplaybreaks

\AtBeginDocument{%
  \providecommand\BibTeX{{%
    \normalfont B\kern-0.5em{\scshape i\kern-0.25em b}\kern-0.8em\TeX}}}

\copyrightyear{2022}
\acmYear{2022}
\setcopyright{acmcopyright}\acmConference[CIKM '22]{Proceedings of the 31st ACM International Conference on Information and Knowledge Management}{October 17--21, 2022}{Atlanta, GA, USA}
\acmBooktitle{Proceedings of the 31st ACM International Conference on Information and Knowledge Management (CIKM '22), October 17--21, 2022, Atlanta, GA, USA}
\acmPrice{15.00}
\acmDOI{10.1145/3511808.3557476}
\acmISBN{978-1-4503-9236-5/22/10}

\begin{document}

\title{Towards Principled User-side Recommender Systems}

%% used to denote shared contribution to the research.+
\author{Ryoma Sato}
\email{r.sato@ml.ist.i.kyoto-u.ac.jp}
\affiliation{%
  \institution{Kyoto University / RIKEN AIP}
  \city{Kyoto}
  \country{Japan}
}

\renewcommand{\shortauthors}{Sato}

\begin{abstract}
Traditionally, recommendation algorithms have been designed for service developers. However, recently, a new paradigm called user-side recommender systems has been proposed and they enable web service users to construct their own recommender systems without access to trade-secret data. This approach opens the door to user-defined fair systems even if the official recommender system of the service is not fair. While existing methods for user-side recommender systems have addressed the challenging problem of building recommender systems without using log data, they rely on heuristic approaches, and it is still unclear whether constructing user-side recommender systems is a well-defined problem from theoretical point of view. In this paper, we provide theoretical justification of user-side recommender systems. Specifically, we see that hidden item features can be recovered from the information available to the user, making the construction of user-side recommender system well-defined. However, this theoretically grounded approach is not efficient. To realize practical yet theoretically sound recommender systems, we propose three desirable properties of user-side recommender systems and propose an effective and efficient user-side recommender system, \textsc{Consul}, based on these foundations. We prove that \textsc{Consul} satisfies all three properties, whereas existing user-side recommender systems lack at least one of them. In the experiments, we empirically validate the theory of feature recovery via numerical experiments. We also show that our proposed method achieves an excellent trade-off between effectiveness and efficiency and demonstrate via case studies that the proposed method can retrieve information that the provider's official recommender system cannot.
\end{abstract}


\begin{CCSXML}
<ccs2012>
   <concept>
       <concept_id>10002951.10003317.10003347.10003350</concept_id>
       <concept_desc>Information systems~Recommender systems</concept_desc>
       <concept_significance>500</concept_significance>
       </concept>
 </ccs2012>
\end{CCSXML}

\ccsdesc[500]{Information systems~Recommender systems}

\keywords{recommender systems; fairness; user-side systems}

\maketitle

\section{Introduction}

Recommender systems have been adopted in various web services \cite{linden2003amazon, geyik2019fairness}, and in particular trustworthy systems have been demanded, including fair \cite{kamishima2012enhancement, biega2018equity, milano2020recommender}, transparent \cite{sinha2002role, balog2019transparent}, and steerable \cite{green2009generating, balog2019transparent} recommender systems. However, the adoption of such trustworthy systems is still limited \cite{sato2022private}, and many users still receive black-box and possibly unfair recommendations. Even if these users are dissatisfied with the system, they have limited power to change recommendations. For example, suppose a Democrat supporter may acknowledge that she is biased towards Democrats, but she wants to get news about Republicans as well. However, a news recommender system may show her only news related to Democrats to maximize the click through rate. In other words, filter bubbles exist in recommender systems \cite{pariser2011filter}. The only recourse available to the user is to wait until the news feed service implements a fair/diverse recommendation engine. \citet{green2009generating} also noted that ``If users are unsatisfied with the recommendations generated by a particular system, often their only way to change how recommendations are generated in the future is to provide thumbs-up or thumbs-down ratings to the system.'' Worse, there are many fairness/diversity criteria, and the service provider cannot address all the criteria. For example, even if the recommendation results are fair with respect to race, some other users may call for fairness with respect to gender. Or some users call for demographic parity when the service realizes the equalized odds.

User-side recommender systems \cite{sato2022private} provide a proactive solution to this problem. The users can build their own (i.e., private, personal, or user-side) recommender systems to ensure that recommendations are generated in a fair and transparent manner. As the system is build by the user, the system can be tailored to meet their own criteria with user-defined protected groups and user-defined criteria. Therefore, user-side recommender systems can be seen as ultimate personalization.

Although this concept is similar to that of steerable \cite{green2009generating} (or scrutable \cite{balog2019transparent}) recommender systems, the difference lies in the fact that steerable systems are implemented by the service provider, whereas user-side recommender systems are implemented by the users. If the recommender system in operation is not steerable, users need to wait for the service provider to implement a steerable system. By contrast, user-side recommender systems enable the users to make the recommender system steerable by their own efforts even if the service provider implemented an ordinary non-steerable system.

Although user-side recommender systems are desirable, the problem of building user-side recommender systems is challenging. In particular, end users do not have access to the trade-secret data stored in the database of the system unlike the developers employed by the service provider. This introduces information-theoretic limitation when users build user-side recommender systems. Most modern recommender systems use users' log data and/or item features to make recommendations. At first glance, it seems impossible to build an effective recommender system without such data. \citet{sato2022private} solved this problem by utilizing official recommender systems in the targeted web service, whose outputs are available yet are black-box and possibly unfair.
Specifically, \citet{sato2022private} proposed two algorithms for user-side recommender systems, \textsc{PrivateRank} and \textsc{PrivateWalk} and achieved empirically promising results. However, they rely on ad-hoc heuristics, and it is still unclear whether building recommender systems without log data is a well-defined problem. In this study, we address this question by using the metric recovery theory for unweighted $k$-nearest neighbor graphs. This result provides a theoretically grounded way to construct user-side recommender systems. However, we found that this approach suffers from large communication costs, and thus is not practical.

Then, we formulate the desirable properties for user-side recommender systems, including minimal communication cost, based on which we propose a method, \textsc{Consul}, that satisfies these properties. This method is in contrast to the existing methods that lack at least one axioms.

The contributions of this study are as follows:
\begin{itemize}
    \item The metric recovery theory validates that curious users can reconstruct the original information of items without accessing log data or item databases and solely from the recommendation results shown in the web page. On the one hand, this result shows that the raw recommendation suggested to the user contains sufficient information to construct user-side recommender systems. This result indicates that user-side recommender systems are feasible in principle. One the other hand, it raises new concerns regarding privacy (Section \ref{sec: recover}).
    \item We list three desirable properties of user-side recommender systems, i.e., consistency, soundness, and locality, and show that existing user-side recommender systems lack at least one of them (Section \ref{sec: principles}).
    \item We propose an algorithm for user-side recommender systems, \textsc{Consul}, which satisfies all three properties, i.e., consistency, soundness, and locality (Section \ref{sec: proposed}).
    \item We empirically validate that personal information can be reconstructed solely from raw recommendation results (Section \ref{sec: experiments-reverse}).
    \item We empirically demonstrate that our proposed method strikes excellent trade-off between effectiveness and communication efficiency (Section \ref{sec: experiments-performance}).
    \item We conduct case studies employing crowd-sourcing workers and show that user-side recommender systems can retrieve novel information that the provider's official recommender system cannot (Section \ref{sec: experiments-user}).
    \item We deploy \textsc{Consul} in the real-world Twitter environment and confirmed that \textsc{Consul} realizes an effective and efficient user-side recommender system. (Section \ref{sec: twitter}).
\end{itemize}

\begin{tcolorbox}[colframe=gray!20,colback=gray!20,sharp corners]
\textbf{Reproducibility}: Our code is available at \url{https://github.com/joisino/consul}.
\end{tcolorbox}

\section{Notations}

For every positive integer $n \in \mathbb{Z}_+$, $[n]$ denotes the set $\{ 1, 2, \dots n \}$.
A lowercase letter, such as $a$, $b$, and $c$, denotes a scalar, and a bold lower letter, such as $\bolda$, $\boldb$, and $\boldc$, denotes a vector. 
Let $\mathcal{I} = [n]$ denote the set of items, where $n$ is the number of items. Without loss of generality, we assume that the items are numbered with $1, \dots, n$. $K \in \mathbb{Z}_+$ denotes the length of a recommendation list. 

\section{Problem Setting} \label{sec: setting}

\begin{table}[tb]
    \centering
    \caption{Notations.}
    \vspace{-0.1in}
    \begin{tabular}{ll} \toprule
        Notations & Descriptions \\ \midrule
        $[n]$ & The set $\{ 1, 2, \dots, n \}$. \\
        $a, \bolda$ & A scalar and vector. \\
        $G = (V, E)$ & A graph. \\
        $\mathcal{I} = [n]$ & The set of items. \\
        $\mathcal{A}$ & The set of protected groups. \\
        $a_i \in \mathcal{A}$ & The protected attribute of item $i \in \mathcal{I}$. \\
        $\mathcal{H} \subset \mathcal{I}$ & The set of items that have been interacted with. \\
        $\mathcal{P}_{\text{prov}}$ & The provider's official recommender system. \\
        $K \in \mathbb{Z}_+$ & The length of a recommendation list. \\
        $\tau \in \mathbb{Z}_{\ge 0}$ & The minimal requirement of fairness. \\
        $d \in \mathbb{Z}_+$ & The number of dimensions of the embeddings. \\
        \bottomrule
    \end{tabular}
    \label{tab: notations}
    \vspace{-0.1in}
\end{table}

\textbf{Provider Recommender System.} In this study, we focus on item-to-item recommendations, following \citet{sato2022private}\footnote{Note that item-to-user recommender systems can also be built based on on item-to-item recommender systems by, e.g., gathering items recommended by an item-to-item recommender system for the items that the user has purchased.}. Specifically, when we visit the page associated with item $i$, $K$ items $\mathcal{P}_\text{prov}(i) \in \mathcal{I}^K$ are presented, i.e., for $k = 1, 2, \cdots, K$, $\mathcal{P}_\text{prov}(i)_k \in \mathcal{I}$ is the $k$-th relevant item to item $i$. For example, $\mathcal{P}_\text{prov}$ is realized in the ``Customers who liked this also liked'' form in e-commerce platforms. We call $\mathcal{P}_\text{prov}$ the \emph{service provider's official recommender system}. We assume that $\mathcal{P}_\text{prov}$ provides relevant items but is unfair and is a black-box system. The goal is to build a fair and white-box recommender system by leveraging the provider's recommender system.

\vspace{0.1in}
\noindent \textbf{Embedding Assumption.} We assume for the provider's recommender system that there exists a representation vector $\boldx_i \in \mathbb{R}^d$ of each item $i$, and that $\mathcal{P}_\text{prov}$ retrieves the top-$k$ nearest items from $\boldx_i$. We do not assume the origin of $\boldx$, and it can be a raw feature vector of an item, a column of user-item interaction matrix, hidden embedding estimated by matrix factorization, or hidden representation computed by neural networks. The embeddings can even be personalized. This kind of recommender system is one of the standard methods \cite{barkan2016item2vec, yao2018judging, sato2022private}. It should be noted that \emph{we assume that $\boldx_i \in \mathbb{R}^d$ cannot be observed} because such data are typically stored in confidential databases of the service provider, and the algorithm for embeddings is typically industrial secrets \cite{milano2020recommender}. For example, $\boldx_i$ may be a column of user-item interaction matrix, but a user cannot observe the interaction history of other users or how many times each pair of items is co-purchased. One of the main argument shown below is that even if $\{\boldx_i\}$ is not directly observable, we can ``reverse-engineer'' $\{\boldx_i\}$ solely from the outputs of the official recommender system $\mathcal{P}_\text{prov}$, which is observable by an end user. Once $\{\boldx_i\}$ is recovered, we can easily construct user-side recommender systems by ourselves using standard techniques \cite{geyik2019fairness, liu2019personalized, zehlike2017fair}.

\vspace{0.1in}
\noindent \textbf{Sensitive Attributes.} We assume that each item $i$ has a discrete sensitive attribute $a_i \in \mathcal{A}$, where $\mathcal{A}$ is the set of sensitive groups. For example, in a talent market service, each item represents a person, and $\mathcal{A}$ can be gender or race. In a news recommender system, each item represents a news article, and $\mathcal{A}$ can be $\{$Republican, Democrat$\}$. 
We want a user-side recommender system that offers items from each group in equal proportions. Note that our algorithm can be extended to any specified proportion (e.g., so that demographic parity holds), but we focus on the equal proportion for simplicity. We assume that sensitive attribute $a_i$ \emph{can be observed}, which is the common assumption in \cite{sato2022private}. Admittedly, this assumption does not necessarily hold in practice. However, when this assumption violates, one can estimate $a_i$ from auxiliary information or recovered embedding $\hat{\boldx}_i$, and the estimation of the attribute is an ordinary supervised learning task and can be solved by off-the-shelf methods, such as neural networks and random forests. As the attribute estimation process is not relevant to the core of user-side recommender system algorithms, this study focuses on the setting where the true $a_i$ can be observed.

\vspace{0.1in}
\noindent \textbf{Other Miscellaneous Assumptions.} We assume that the user knows the set $\mathcal{H} \subset \mathcal{I}$ of items that he/she has already interacted with. This is realized by, for instance, accessing the purchase history page. In addition, we assume that $\mathcal{P}_\text{prov}$ does not recommend items in $\mathcal{H}$ or duplicated items. These assumptions are optional and used solely for technical reasons. If $\mathcal{H}$ is not available, our algorithm works by setting $\mathcal{H}$ as the empty set.

\vspace{0.1in}
\noindent The problem setting can be summarized as follows:

\begin{tcolorbox}[colframe=gray!20,colback=gray!20,sharp corners]
\noindent \uline{\textbf{User-side Recommender System Problem.}}\\
\textbf{Given:} Oracle access to the official recommendations $\mathcal{P}_{\text{prov}}$. Sensitive attribute $a_i$ of each item $i \in \mathcal{I}$.\\
\textbf{Output:} A user-side recommender system $\mathcal{Q}\colon \mathcal{I} \to \mathcal{I}^K$ that is fair with respect to $\mathcal{A}$.
\end{tcolorbox}

\section{Interested Users Can Recover Item Embeddings} \label{sec: recover}

We first see that a user without access to the database can recover the item embeddings solely from the recommendation results $\mathcal{P}_\text{prov}$.

\vspace{0.1in}
\noindent \textbf{Recommendation Networks.} A recommendation network is a graph where nodes represent items and edges represent recommendation relations. Recommendation networks have been traditionally utilized to investigate the properties of recommender systems \cite{cano2006topology, celma2008new, seyerlehner2009limitation}. They were also used to construct user-side recommender systems \cite{sato2022private}. Recommendation network $G = (V, E)$ we use in this study is defined as follows:
\begin{itemize}
    \item Node set $V$ is the item set $\mathcal{I}$.
    \item Edge set $E$ is defined by the recommendation results of the provider's recommender system. There exists a directed edge from $i \in V$ to $j \in V$ if item $j$ is included in the recommendation list in item $i$, i.e., $\exists k \in [K] \text{ s.t. } \mathcal{P}_\text{prov}(i)_k = j$.
    \item We do not consider edge weights. 
\end{itemize}
It should be noted that $G$ can be constructed solely by accessing $\mathcal{P}_\text{prov}$. In other words, an end user can observe $G$. In practice, one can crawl the web pages to retrieve top-$k$ recommendation lists. We use unweighted recommendation networks because weights between two items, i.e., the similarity score of two items, are typically \emph{not} presented to users but only a list of recommended items is available in many web services. This setting makes it difficult to estimate the hidden relationship from the recommendation network. Although the original data that generate the recommendation results, e.g., how many times items $i$ and $j$ are co-purchased, cannot be observed, the recommendation results provide useful information on the similarity of items. A critical observation is that, from the embedding assumption we introduce in Section \ref{sec: setting}, graph $G$ can be seen as the $k$-nearest neighbor ($k$-NN) graph of hidden item embeddings $\{\boldx_i\}$.

\vspace{0.1in}
\noindent \textbf{Embedding Recovery.} We discuss that the original embeddings can be recovered from the unweighted recommendation network based on the metric recovery theory \cite{hashimoto2015metric,luxburg2013density,terada2014local}.

First, it is impossible to exactly recover the original embeddings, but there exist some degrees of freedom such as translation, rotation, reflections, and scaling because such similarity transformations do not change the $k$-NN graph. However, similarity transformations are out of our interest when we construct recommender systems, and fortunately, the original embeddings can be recovered up to similarity transformations.

\citet{hashimoto2015metric} considered metric recovery from the $k$-NN graph of a general point cloud $\{\boldx_i\} \subset \mathbb{R}^d$ assuming $\{\boldx_i\}$ is sampled from an unknown distribution $p(x)$ and showed that under mild assumptions of the embedding distribution, the stationary distribution of the simple random walk on the $k$-NN graph converges to $p(x)$ with appropriate scaling (Corollary 2.3 in \cite{hashimoto2015metric}) when $k = \omega(n^{\frac{2}{d+2}} \log^{\frac{d}{d+2}} n)$. \citet{alamgir2012shortest} showed that by setting the weight of the $k$-NN graph $G$ to the scaled density function, the shortest path distance on $G$ converges to the distance between the embeddings. Therefore, distance matrix $D \in \mathbb{R}^{n \times n}$ of $\{\boldx_1, \cdots, \boldx_n\}$ can be estimated from the $k$-NN graph (Theorem S4.5 in \cite{hashimoto2015metric}). From the standard argument on the multi-dimensional scaling \cite{sibson1979studies}, we can estimate the original embeddings $\{\boldx_1, \cdots, \boldx_n\}$ from the distance matrix up to similarity transformations with sufficient number of samples $n$. 

\citet{terada2014local} also considered embedding recovery from unweighted $k$-NN graphs\footnote{Note that the use of the these methods require knowledge of the dimension $d$ of the embeddings. However, when the dimensionality is not available, it can be estimated from the unweighted $k$-NN graph using the estimator proposed in \cite{kleindessner2015dimensionality}. For the sake of simplicity, we assume that the true number of dimensions $d$ is known.}. Theorem 3 in \cite{terada2014local} shows that their proposed method, local ordinal embeddings (LOE), recovers the original embeddings $\{\boldx_i\}$ from the $k$-NN graph up to similarity transformations. Although the original proof of Theorem 3 in \cite{terada2014local} relied on the result of \cite{luxburg2013density}, which is valid only for $1$-dimensional cases, the use of Corollary 2.3 in \cite{hashimoto2015metric} validates the LOE for general dimensions. The strength of LOE is that it can estimate the embeddings in a single step while the density estimation approach \cite{luxburg2013density, hashimoto2015metric} requires intermediate steps of estimating the density function and distance matrix. We use the LOE in our experiments.

In summary, the item embeddings $\{\boldx_i\}$ can be recovered solely from the recommendation network. Once embeddings $\{\boldx_i\}$ are recovered, the user has access to the same information as the service provider, after which it is straightforward to construct user-side recommender systems, e.g., by adopting a post-processing fair recommendation algorithm \cite{geyik2019fairness, liu2019personalized, zehlike2017fair}. We call this approach estimate-then-postprocessing (ETP).

These results demonstrate that the $k$-NN recommendation network, which is observable by an end user, contains sufficient information for constructing user-side recommender systems. Therefore, from the perspective of information limit, the user-side recommender system problem is feasible.

To the best of our knowledge, this work is the first to adopt the embedding recovery theory to user-side recommender systems.

\vspace{0.1in}
\noindent \textbf{Security of Confidential Data.} Although it was not our original intent, the discussion above reveals a new concern regarding data privacy, i.e., confidential information about items may be unintentionally leaked from the recommender systems. For example, in a talent search service, a recommender system may use personal information, such as educational level, salary, and age, to recommend similar talents. This indicates that we can recover such data solely from the recommendation results, which we numerically confirm in the experiments. Therefore, service providers need to develop rigorous protection techniques to prevent leaks of personal information from the recommender system. This is beyond the scope of this work, and we leave this interesting problem as future work.
 
\vspace{0.1in}
\noindent \textbf{Limitation of ETP.} Although ETP is conceptually sound and provides a theoretical foundation of user-side recommender systems, it is impractical because it incurs high communication costs. LOE requires observing the entire $k$-NN graph, and so do other methods, including the path-based approach \cite{luxburg2013density} and global random walk approach \cite{hashimoto2015metric}. Indeed, \citet{luxburg2013density} noted that ``It is impossible to estimate the density in an unweighted $k$-NN graph by local quantities alone.'' As the density estimation can be reduced to embedding recovery, it is also impossible to recover the embedding by local quantities of an unweighted $k$-NN graph. This means that ETP requires to download all item pages from the web service to construct user-side recommender systems, which incurs prohibitive communication costs when the number of items is large. From a practical point of view, it is not necessary to exactly recover the ground truth embeddings. It may be possible to derive useful user-side recommender systems by bypassing the reverse engineering. We will formulate the three desiderata and propose a method that provably satisfies the desiderata in the following sections.

\section{Design Principles} \label{sec: principles}

We formulate three design principles for user-side recommender systems in this section. We consider user-side recommender systems that take a parameter $\tau \in \mathbb{Z}_{\ge 0}$ to control the trade-off between fairness and performance. $\tau = 0$ indicates that the recommender system does not care about fairness, and an increase in $\tau$ should lead to a corresponding increase in fairness.

\vspace{0.1in}
\noindent \textbf{Consistency.} A user-side recommender system $\mathcal{Q}$ is consistent if nDCG of $\mathcal{Q}$ with $\tau = 0$ is guaranteed to be the same as that of the official recommender system. In other words, a consistent user-side recommender system does not degrade the performance if we do not impose the fairness constraint. \citet{sato2022private} showed that \textsc{PrivateRank} is consistent.

\vspace{0.1in}
\noindent \textbf{Soundness.} We say a user-side recommender system is sound if the minimum number of items from each sensitive group is guaranteed to be at least $\tau$ provided $0 \le \tau \le K/|\mathcal{A}|$ and there exist at least $\tau$ items for each group $a \in \mathcal{A}$\footnote{Note that if $\tau > K/|\mathcal{A}|$ or there exist less than $\tau$ items for some group $a$, it is impossible to achieve this condition.}. In other words, we can provably guarantee the fairness of sound user-side recommender system by adjusting $\tau$. \citet{sato2022private} showed that \textsc{PrivateRank} is sound.

\vspace{0.1in}
\noindent \textbf{Locality.} We say a user-side recommender system is local if it generates recommendations without loading the entire recommendation network. This property is crucial for communication costs. As we mentioned above, the ETP approach does not satisfy this property.

\vspace{0.1in}
\noindent \textbf{Inconsistency of \textsc{PrivateWalk}.} The following proposition shows that \textsc{PrivateWalk} is not consistent \footnote{This fact is also observed from the experimental results in \cite{sato2022private}.}

\begin{proposition}
\textsc{PrivateWalk} is not consistent.
\end{proposition}

\begin{proof}
We construct a counterexample where \textsc{PrivateWalk} with $\tau = 0$ generates different recommendations from the provider. Let $\mathcal{I} = [5]$ and $K = 2$,
\begin{align*}
\mathcal{P}_\text{prov}(i)_1 &= ((i + 3) \text{ mod } 5) + 1, \\
\mathcal{P}_\text{prov}(i)_2 &= (i \text{ mod } 5) + 1.
\end{align*}
The recommendations of the provider in item $3$ is $(2, 4)$. Let \textsc{PrivateWalk} select the first item (item $2$) as a result of randomness. Then, \textsc{PrivateWalk} accepts item $2$ and restarts the random walk. Let \textsc{PrivateWalk} select the first item (item $2$) again. \textsc{PrivateWalk} rejects item $2$ and continues the random walk. Let \textsc{PrivateWalk} select the first item (item $1$). Then, \textsc{PrivateWalk} accepts item $1$ and terminates the process. The final recommendation result is $(2, 1)$, which is different from the recommendations of the provider. If item $4$ is the ground truth item, the nDCG of \textsc{PrivateWalk} is less than the official recommender system.
\end{proof}

\begin{table}[tb]
\small
    \caption{Properties of user-side recommender systems. Postprocessing (PP) applies postprocessing directly to the official recommender system, which is not sound when the list does not contain some sensitive groups \cite{sato2022private}.}
    \vspace{-0.1in}
    \centering
    \begin{tabular}{lccccc} \toprule
    & PP & \textsc{PrivateRank} & \textsc{PrivateWalk} & ETP & \textsc{CONSUL} \\ \midrule
    Consistent & \cmark & \cmark & \xmark & \cmark & \cmark \\
    Sound & \xmark & \cmark & \cmark & \cmark & \cmark \\
    Local & \cmark & \xmark & \cmark & \xmark & \cmark \\ \bottomrule
    \end{tabular}
    \label{tab: prop}
\end{table}

The properties of user-side recommender systems are summarized in Table \ref{tab: prop}. The existing methods lack at least one properties. We propose a method that satisfies all properties in the next section.

\section{Proposed Method} \label{sec: proposed}

\setlength{\textfloatsep}{5pt}
\begin{algorithm2e}[t]
\caption{\textsc{Consul}}
\label{algo: consul}
\DontPrintSemicolon 
\nl\KwData{Oracle access to $\mathcal{P}_{\text{prov}}$, Source item $i \in \mathcal{I}$, Protected attributes $a_i ~\forall i \in \mathcal{I}$, Minimum requirement $\tau$, Set $\mathcal{H}$ of items that user $i$ has already interacted with, Maximum length $L_{\text{max}}$ of search.}
\nl\KwResult{Recommended items $\mathcal{R} = \{j_k\}_{1 \le k \le K}$.}
    \nl Initialize $\mathcal{R} \leftarrow []$ (empty), $p \leftarrow i$ \;
    \nl $c[a] \leftarrow 0 \quad \forall a \in \mathcal{A}$ \tcp*{counter of sensitive groups}
    \nl $\mathcal{S} \leftarrow \text{Stack}([])$ \tcp*{empty stack}
    \nl \For{\textup{iter} $\gets 1$ \textbf{to} $L_{\text{max}}$}{
    \nl     \While{$p$ \textup{is already visited}}{

    \nl     \If{$|\mathcal{S}| = 0$}{
    \nl         \textbf{goto} line 21 \tcp*{no further items}
            }
    \nl     $p \leftarrow \mathcal{S} \text{.pop\_top}()$  \tcp*{next search node}
            }
    \nl     \For{$k \gets 1$ \textbf{to} $K$}{
    \nl         $j \leftarrow \mathcal{P}_{\text{prov}}(p)_{k}$ \;
    \nl         \If{$j \textup{ \textbf{not} \textbf{in} } \mathcal{R} \cup \mathcal{H}$ \textup{\textbf{and}} $\sum_{a \neq a_j} \max(0, \tau - c[a]) \le K - |\mathcal{R}| - 1$}{
    \nl             \tcc{$j$ can be safely added keeping fairness. Avoid items in $\mathcal{R} \cup \mathcal{H}$.}
    \nl             Push back $j$ to $\mathcal{R}$. \;
    \nl             $c[a_j] \leftarrow c[a_j] + 1$ \;
                }
    \nl         \If{$|\mathcal{R}| = K$}{
    \nl             \textbf{return} $\mathcal{R}$ \tcp*{list is full}
                }
            }
    \nl     \For{$k \gets K$ \textbf{to} $1$}{
    \nl         $\mathcal{S} \text{.push\_top}(\mathcal{P}_{\text{prov}}(p)_{k})$ \tcp*{insert candidates}
            }
        }
    \nl \While{$|\mathcal{R}| < K$}{
    \nl     $j \leftarrow$ \text{Uniform}($\mathcal{I}$) \tcp*{random item}
    \nl     \If{$j \textup{ \textbf{not} \textbf{in} } \mathcal{R} \cup \mathcal{H}$ \textup{\textbf{and}} $\sum_{a \neq a_j} \max(0, \tau - c[a]) \le K - |\mathcal{R}| - 1$}{
    \nl         Push back $j$ to $\mathcal{R}$. \;
    \nl         $c[a_j] \leftarrow c[a_j] + 1$ \;
            }
        }
    \nl \textbf{return} $\mathcal{R}$ \;
\end{algorithm2e}

In this section, we propose \textbf{CON}sistent, \textbf{S}o\textbf{U}nd, and \textbf{L}ocal user-side recommender system, \textsc{Consul}. \textsc{Consul} inherits the basic idea from \textsc{PrivateWalk}, i.e., close items on the recommendation network of the official recommender system should be similar, and \textsc{Consul} retrieves similar items using a walk-based algorithm. However, there are several critical differences between \textsc{Consul} and \textsc{PrivateWalk}. First, \textsc{Consul} runs a deterministic depth-first search instead of random walks. This reduces the variance of the algorithm by removing the stochasticity from the algorithm. Second, \textsc{Consul} uses all recommended items when it visits an item page, whereas \textsc{PrivateWalk} chooses only one of them. This feature improves efficiency by filling the recommendation list with short walk length. Third, \textsc{Consul} continues the search after it finds a new item, whereas \textsc{PrivateWalk} restarts a random walk. When items of the protected group are distant from the source node, \textsc{PrivateWalk} needs to travel long distances many times. By contrast, \textsc{Consul} saves such effort through continuation. Although each modification is small, the combined improvements make a considerable difference to the theoretical properties, as we show in Theorem \ref{thm: consul}, and empirical performances, as we will show in the experiments. We stress that the simplicity of \textsc{Consul} is one of the strengths of our proposed method.

\vspace{0.1in}
\noindent \textbf{Pseudo Code.} Algorithm \ref{algo: consul} shows the pseudo code. Lines 3--5 initialize the variables. In lines 11--18, items included in the recommendation list in item $p$ are added to the list as long as the insertion is safe. The condition $\sum_{a \neq a_j} \max(0, \tau - c[a]) \le K - |\mathcal{R}| - 1$ ensures soundness as we will show in the theorem below. In lines 19--20, the adjacent items are included in the search stack in the descending order. This process continues until $K$ items are found, $L_{\text{max}}$ nodes are visited, or no further items are found. In lines 21--25, the fallback process ensures that $K$ items are recommended. Note that this fallback process is usually skipped because $K$ items are found in the main loop.

\begin{theorem} \label{thm: consul}
\textsc{Consul} is consistent, sound, and local.
\end{theorem}

\begin{figure*}[tb]
\centering
\includegraphics[width=0.95\hsize]{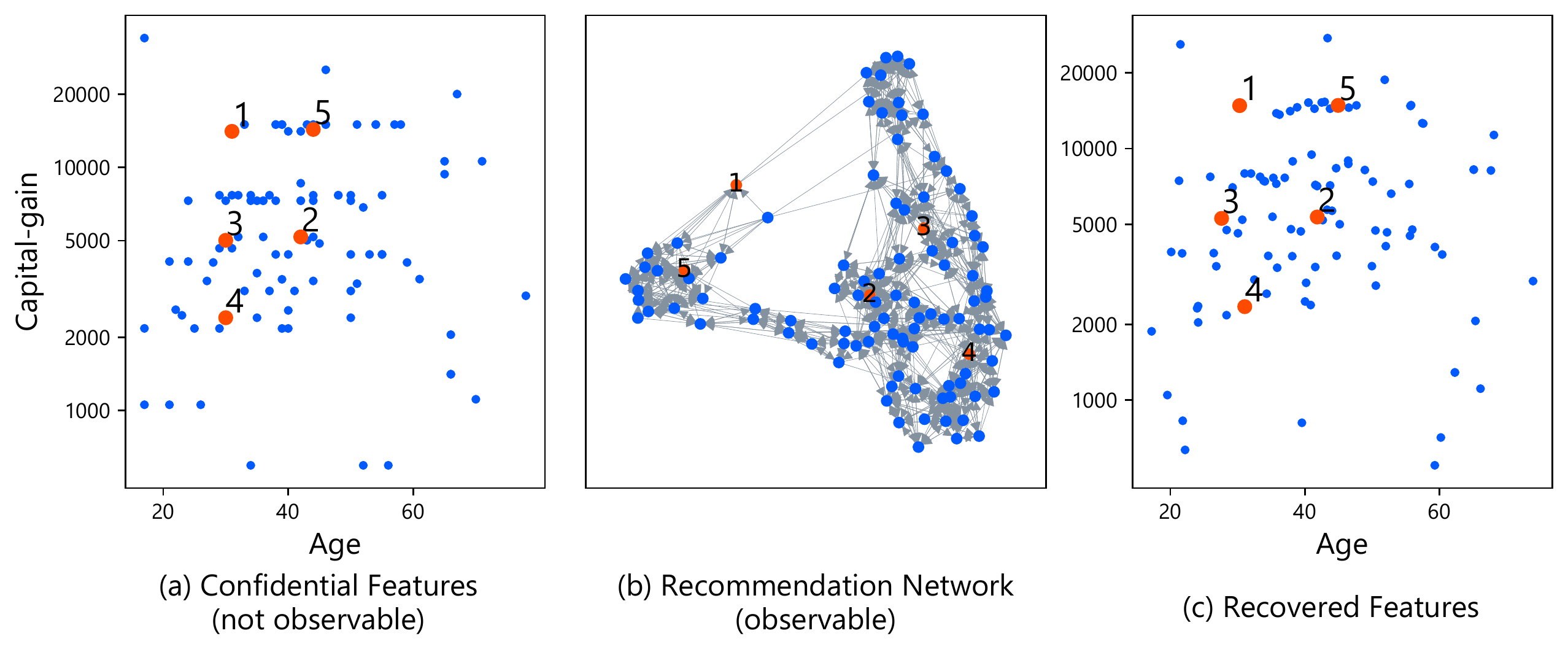}
\vspace{-0.2in}
\caption{\textbf{Feature Reverse Engineering.} \textbf{(Left)} The original features which are confidential. \textbf{(Middle)} The $k$-NN recommendation network revealed to the user. We visualize the graph with \texttt{networkx} package. \textbf{(Right)} The recovered feature solely from the recommendation results. Items $1$ to $5$ are colored red for visibility. All features are accurately recovered. These results show that the raw recommendation results contain sufficient information to build user-side recommender systems.}
\label{fig: reverse}
\vspace{-0.1in}
\end{figure*}

\begin{proof}
\noindent \textbf{Consistency.} If $\tau = 0$, $\sum_{a \neq a_j} \max(0, \tau - c[a]) = 0$ holds in line 13. Therefore, the condition in line 13 passes in all $K$ iterations in the initial node, and the condition in line 17 passes at the $K$-th iteration. The final output is $\mathcal{R} = \mathcal{P}_{\text{prov}}(i)$.

\vspace{0.1in}
\noindent \textbf{Soundness.} We prove by mathematical induction that \begin{align}\sum_{a \in \mathcal{A}} \max(0, \tau - c[a]) \le K - |\mathcal{R}| \label{eq: consul-proof-eq}\end{align} holds in every step in Algorithm \ref{algo: consul}. At the initial state, $|\mathcal{R}| = 0$, $c[a] = 0 ~\forall a \in \mathcal{A}$, and $\sum_{a \in \mathcal{A}} \max(0, \tau - c[a]) = \tau |\mathcal{A}|$. Thus, the inequality holds by the assumption $\tau |\mathcal{A}| \le K$. The only steps that alter the inequality are lines 15--16 and lines 24-25. Let $c, \mathcal{R}$ be the states before the execution of these steps, and $c', \mathcal{R}'$ be the states after the execution. We prove that $\sum_{a \in \mathcal{A}} \max(0, \tau - c'[a]) \le K - |\mathcal{R}'|$ holds assuming \begin{align}\sum_{a \in \mathcal{A}} \max(0, \tau - c[a]) \le K - |\mathcal{R}|, \label{eq: consul-proof-inductive}\end{align} i.e., the inductive hypothesis.  When these steps are executed, the condition \begin{align}\sum_{a \neq a_j} \max(0, \tau - c[a]) \le K - |\mathcal{R}| - 1 \label{eq: consul-proof-if}\end{align} holds by the conditions in lines 13 and 23. We consider two cases. (i) If $c[a_j] \ge \tau$ holds,
\begin{align*}
\sum_{a \in \mathcal{A}} \max(0, \tau - c'[a])
&\stackrel{\text{(a)}}{=} \sum_{a \neq a_j} \max(0, \tau - c'[a]) \\
&\stackrel{\text{(b)}}{=} \sum_{a \neq a_j} \max(0, \tau - c[a]) \\
&\stackrel{\text{(c)}}{\le} K - |\mathcal{R}| - 1 \\
&= K - |\mathcal{R}'|,
\end{align*}
where (a) follows $c[a_j] \ge \tau$, (b) follows $c'[a] = c[a] ~(\forall a \neq a_j)$, and (c) follows eq. \eqref{eq: consul-proof-if}. (ii) If $c[a_j] < \tau$ holds,
\begin{align*}
\sum_{a \in \mathcal{A}} \max(0, \tau - c'[a])
&\stackrel{\text{(a)}}{=} \sum_{a \in \mathcal{A}} \max(0, \tau - c[a]) - 1 \\
&\stackrel{\text{(b)}}{\le} K - |\mathcal{R}| - 1 \\
&= K - |\mathcal{R}'|,
\end{align*}
where (a) follows $c'[a_j] = c[a_j] + 1$ and $c[a_j] + 1 \le \tau$, and (b) follows eq. \eqref{eq: consul-proof-inductive}. In sum, eq. \eqref{eq: consul-proof-eq} holds by mathematical induction. When Algorithm \ref{algo: consul} terminates, $|\mathcal{R}| = K$. As the left hand side of eq. \eqref{eq: consul-proof-eq} is non-negative, each term should be zero. Thus, $c[a] = |\{i \in \mathcal{R} \mid a_i = a\}| \ge \tau$ holds for all $a \in \mathcal{A}$.

\vspace{0.1in}
\noindent \textbf{Locality.} \textsc{Consul} accesses the official recommender system in lines 12 and 20. As the query item $p$ changes at most $L_\text{max}$ times in line 10, \textsc{Consul} accesses at most $L_\text{max}$ items, which is a constant, among $n$ items.
\end{proof}

\vspace{0.1in}
\noindent \textbf{Time complexity.} The time complexity of \textsc{Consul} depends on the average length $L_{\text{ave}}$ of random walks, which is bounded by $L_{\text{max}}$. The number of loops in lines 11--18 is $O(K L_{\text{ave}})$. Although the condition in line 13 involves $|\mathcal{A}|$ terms, it can be evaluated in constant time by storing \[s \stackrel{\text{def}}{=} \sum_{a \in \mathcal{A}} \max(0, \tau - c[a])\] because \[\sum_{a \neq a_j} \max(0, \tau - c[a]) = s - \max(0, \tau - c[a_j]),\] the right hand side of which can be computed in constant time. The number of loops in lines 19-20 is also $O(K L_{\text{ave}})$. Therefore, the main loop in lines 6--20 runs in $O(K L_{\text{ave}})$ time if we assume evaluating $\mathcal{P}_{\text{prov}}$ runs in constant time. Assuming the proportion of each sensitive group is within a constant, the number of loops of the fallback process (lines 21--25) is $O(K |\mathcal{A}|)$ in expectation because the number of trials until the condition in line 23 passes is $O(|\mathcal{A}|)$. Therefore, the overall time complexity is $O(K (|\mathcal{A}| + L_{\text{ave}}))$ and is independent of the number $n$ of items. This is $K$ times faster than \textsc{PrivateWalk}. Besides, in practice, the communication cost is more important than the time complexity. The communication cost of \textsc{Consul} is $L_\text{ave}$, which is $K$ times faster than $O(K L_\text{ave})$ of \textsc{PrivateWalk}. We confirm the efficiency of \textsc{Consul} in experiments.

\section{Experiments} \label{sec: experiments}

We answer the following questions through the experiments.

\begin{itemize}
    \item (RQ1) Can users recover the confidential information solely from the recommendation network?
    \item (RQ2) How good a trade-off between performance and efficiency does \textsc{Consul} strike?
    \item (RQ3) Can user-side recommender systems retrieve novel information?
    \item (RQ4) Does \textsc{Consul} work in the real world?
\end{itemize}

\begin{table*}[t]
    \caption{Performance Comparison. Access denotes the average number of times each method accesses item pages, i.e., the number of queries to the official recommender systems. The less this value is, the more communication-efficient the method is. The best score is highlighted with bold. \textsc{Consul} is extremely more efficient than other methods while it achieves on par or slightly worse performances than Oracle and \textsc{PrivateRank}.}
    \vspace{-0.1in}
    \centering
\begin{tabular}{lcccccccc} \toprule
    & \multicolumn{2}{c}{Adult} & \multicolumn{3}{c}{MovieLens (oldness)} & \multicolumn{3}{c}{MovieLens (popularity)} \\
    \cmidrule(lr{1.0em}){2-3} \cmidrule(lr{1.0em}){4-6} \cmidrule(lr{1.0em}){7-9}
    & Accuracy $\uparrow$ & Access $\downarrow$ & nDCG $\uparrow$ & Recall $\uparrow$ & Access $\downarrow$ & nDCG $\uparrow$ & Recall $\uparrow$ & Access $\downarrow$ \\ \midrule
    Oracle & \textbf{0.788} & $\infty$ & \textbf{0.0321} & \textbf{0.057} & $\infty$ & \textbf{0.034} & \textbf{0.064} & $\infty$ \\
    \textsc{PrivateRank} & 0.781 & 39190 & 0.0314 & 0.055 & 1682 & \textbf{0.034} & 0.062 & 1682 \\
    \textsc{PrivateWalk} & 0.762 & 270.5 & 0.0273 & 0.049 & 154.2 & 0.029 & 0.054 & 44.0 \\
    \textsc{Consul} & 0.765 & \textbf{34.5} & \textbf{0.0321} & \textbf{0.057} & \textbf{19.6} & 0.033 & 0.060 & \textbf{4.6} \\ \bottomrule
\end{tabular}
\newline
\vspace*{0.1 in}
\newline
\begin{tabular}{lcccccc} \toprule
    & \multicolumn{3}{c}{Amazon} & \multicolumn{3}{c}{LastFM} \\
    \cmidrule(lr{1.0em}){2-4} \cmidrule(lr{1.0em}){5-7}
    & nDCG $\uparrow$ & Recall $\uparrow$ & Access $\downarrow$ & nDCG $\uparrow$ & Recall $\uparrow$ & Access $\downarrow$ \\ \midrule
    Oracle & \textbf{0.0326} & \textbf{0.057} & $\infty$ & \textbf{0.0652} & \textbf{0.111} & $\infty$ \\
    \textsc{PrivateRank} & 0.0325 & \textbf{0.057} & 1171 & 0.0641 & 0.107 & 1507 \\
    \textsc{PrivateWalk} & 0.0217 & 0.048 & 135.7 & 0.0424 & 0.080 & 74.1 \\
    \textsc{Consul} & 0.0310 & 0.052 & \textbf{7.3} & 0.0639 & 0.107 & \textbf{6.5} \\ \bottomrule
\end{tabular}
    \label{tab: performance}
    \vspace{-0.1in}
\end{table*}

\subsection{(RQ1) Feature Reverse Engineering} \label{sec: experiments-reverse}

\noindent \textbf{Setup.} We use Adult dataset \cite{dua2017uci}\footnote{\url{https://archive.ics.uci.edu/ml/datasets/adult}} in this experiment. This dataset contains demographic data such as age, sex, race, and income. Although this is originally a census dataset, we use this with talent search in mind. Specifically, we regard a record of a person as an item and construct a provider's recommender system that recommends similar people on each person's page. We use age and capital-gain, which can be confidential information, as features. We remove the highest and lowest values as outliers because they are clipped in the dataset. We take a logarithm for capital-gain because of its high dynamic range. We normalize the features and recommend $K$-nearest neighbor people with respect to the Euclid distance of the features. We recover the features solely from the recommendation results using LOF \cite{terada2014local}.

\vspace{0.1in}
\noindent \textbf{Results.} The left panel of Figure \ref{fig: reverse} visualizes the original data, which are not revealed to us. The middle panel of Figure \ref{fig: reverse} is the observed $k$-NN recommendation network, which is the only information we can use in this experiment. The right panel of Figure \ref{fig: reverse} is the recovered features. We apply the optimal similarity transformation that minimizes the L2 distance to the original embedding for visibility. We also highlight items $1, 2, 3, 4,$ and $5$ in red. These results show that the confidential information is accurately recovered. These results validate that the recommendation network, which is observable by an end user, contains sufficient information for building user-side recommender systems.

\vspace{0.1in}
\noindent \textbf{Discussion on the Similarity Transformation.} As we discussed in Section \ref{sec: recover}, the similarity transformation factor (e.g., rotation) cannot be determined solely from the $k$-NN graph. When we need to recover exact information, we may be able to determine the similarity transformation factor by registering a few items with known features to the item database, e.g., signing up at the talent search service. As the degree of freedom is at most $6$ dimensions in the case of two features, a few items suffice. As the exact recover is not the original aim of this study, we leave exploring more realistic approaches for recovering the exact information as future work.

\subsection{(RQ2) Performance} \label{sec: experiments-performance}

\noindent \textbf{Setup.} We use Adult dataset, MovieLens100k \cite{harper2016movielens}, Amazon Home and Kitchen \cite{he2016ups, mcauley2015image}, and LastFM \footnote{\url{https://grouplens.org/datasets/hetrec-2011/}} datasets following \citep{sato2022private}.

\noindent \textbf{Adult dataset.} In this dataset, an item represents a person, and the sensitive attribute is defined by sex. We use the nearest neighbor recommendations with demographic features, including age, education, and capital-gain, as the provider's official recommender system. The label of an item (person) represents whether the income exceeds \$50\,000 per year. The accuracy for item $i$ represents the ratio of the recommended items for item $i$ that have the same label as item $i$. The overall accuracy is the average of the accuracy of all items.

\noindent \textbf{MovieLens dataset.} In this dataset, an item represents a movie. We consider two ways of creating protected groups, (i) oldness: We regard movies released before 1990 as a protected group, and (ii) popularity: We regard movies with less than $50$ reviews as the protected group. We use Bayesian personalized ranking (BPR) \cite{rendle2009bpr} for the provider's recommender system, where the similarity of items is defined by the inner product of the latent vectors of the items, and the top-$K$ similar items are recommended. We use the default parameters of Implicit package\footnote{\url{https://github.com/benfred/implicit}} for BPR. We measure nDCG@$K$ and recall@$k$ as performance metrics following previous works \cite{krichene2020sampled, he2017neural, rendle2009bpr, sato2022private}. Note that we use the full datasets to compute nDCG and recall instead of employing negative samples to avoid biased evaluations \cite{krichene2020sampled}.

\noindent \textbf{LastFM and Amazon dataset.} In these datasets, an item represents a music and a product, respectively. We regard items that received less than $50$ interactions as a protected group. We extract $10$-cores for these datasets by iteratively discarding items and users with less that $10$ interactions. We use BPR for the provider's official recommender system, as on the MovieLens dataset. We use nDCG@$K$ and recall@$k$ as performance metrics.

In all datasets, we set $K = 10$ and $\tau = 5$, i.e., recommend $5$ protected items and $5$ other items. Note that all methods, \textsc{Consul} and the baselines, are guaranteed to generate completely balanced recommendations, i.e., they recommend $5$ protected items and $5$ other items, when we set $K = 10$ and $\tau = 5$. We checked that the results were completely balanced. So we do not report the fairness scores but only report performance metrics in this experiment. 

\begin{figure}[tb]
\centering
\includegraphics[width=0.8\hsize]{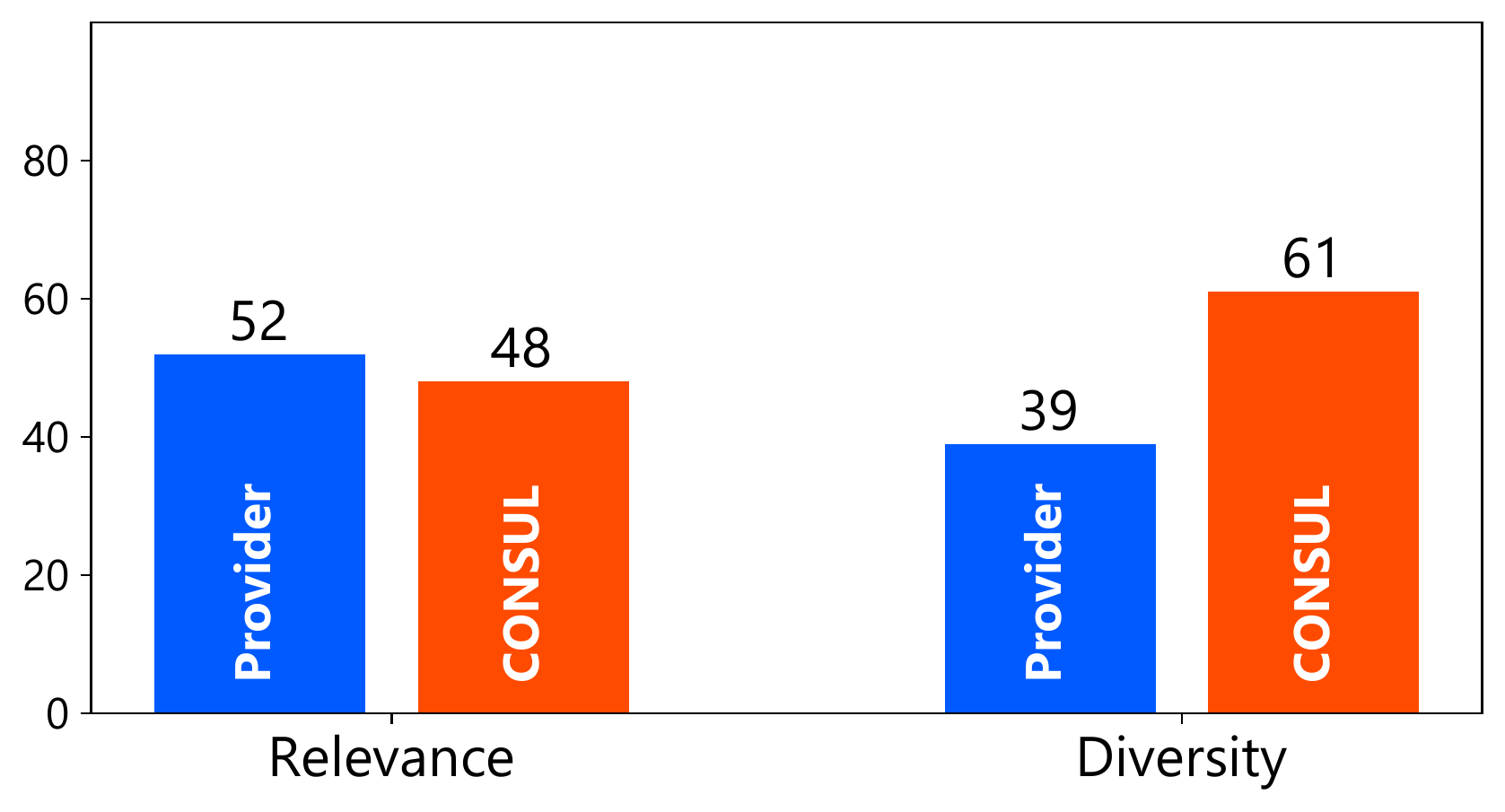}
\vspace{-0.17in}
\caption{\textbf{User Study}. Each value shows the number of times the method is chosen by the crowd workers in $100$ comparisons. Relevance: The crowd workers are asked which recommendation is more relevant to the source item. Diversity: The crowd workers are asked which recommendation is more diverse. \textsc{Consul} is on par or slightly worse than the provider in terms of relevance, while it clearly has more diversity.}
\label{fig: user}
\vspace{-0.03in}
\end{figure}

\vspace{0.1in}
\noindent \textbf{Methods.} We use \textsc{PrivateRank} and \textsc{PrivateWalk} \cite{sato2022private} as baseline methods. We use the default hyperparameters of \textsc{PrivateRank} and \textsc{PrivateWalk} reported in \cite{sato2022private}. Note that \citet{sato2022private} reported that \textsc{PrivateRank} and \textsc{PrivateWalk} were insensitive to the choice of these hyperparameters. In addition, we use the oracle method \cite{sato2022private} that uses the complete similarity scores used in the provider's recommender system and adopts the same fair post-processing as \textsc{PrivateRank}. The oracle method uses hidden information, which is not observable by end users, and can be seen as an ideal upper bound of performance. We tried to apply ETP but found that LOE \cite{terada2014local} did not finish within a week even on the smallest dataset, MovieLens 100k. As ETP is at most as communication efficient as \textsc{PrivateRank}, we did not pursue the use of ETP further in this experiment.

\vspace{0.1in}
\noindent \textbf{Results.} Table \ref{tab: performance} reports performances and efficiency of each method. First, \textsc{Consul} is the most efficient method and requires $10$ times fewer queries to the official recommender system compared with \textsc{PrivateWalk}. Note that \textsc{PrivateWalk} was initially proposed as an efficient method in \cite{sato2022private}. Nevertheless, \textsc{Consul} further improves communication cost by a large margin. Second, \textsc{Consul} performs on par or slightly worse than Oracle and \textsc{PrivateRank}, whereas \textsc{PrivateWalk} degrades performance in exchange for its efficiency. In sum, \textsc{Consul} strikes an excellent trade-off between performance and communication costs.

\vspace{0.1in}
\noindent \textbf{Discussion on the Real-time Inference.} It is noteworthy that \textsc{Consul} makes recommendations with only $6$ to $7$ accesses to item pages in Amazon and FastFM. This efficiency enables it to recommend items in real-time when the user visits an item page. This property of \textsc{Consul} is important because nowadays many recommendation engines are very frequently updated, if not in real time, and maintaining the recommendation graph is costly. \textsc{Consul} can make a prediction on the fly without managing the recommendation graph offline. By contrast, \textsc{PrivateWalk} requires close to 100 accesses, which prohibit real-time inference and thus requires to manage the recommendation graph offline. This is a clear advantage of \textsc{Consul} compared to the existing methods.

\begin{table}[t]
    \caption{Recommendations for ``Terminator 2: Judgment Day.'' The common items are shown in black, and the differences are shown in blue and red for visibility. In this example, \textsc{Consul} was chosen by the crowd workers in terms of \emph{both} relevance and diversity. \textsc{Consul} retrieves classic sci-fi movies in addition to contemporary sci-fi movies while Provider recommends only contemporary movies.}
    \centering
    \vspace{-0.1in}
    \scalebox{0.95}{
    \begin{tabular}{cc} \toprule
        \multicolumn{2}{c}{Terminator 2: Judgment Day (1991)} \\ \midrule
        Provider & \textsc{Consul} \\ \midrule
        Total Recall (1990) & Total Recall (1990) \\
        The Matrix (1990) & The Matrix (1990)  \\
        The Terminator (1990) & The Terminator (1990) \\
        {\color[HTML]{005AFF} Jurassic Park (1993)} & {\color[HTML]{FF4B00} Alien (1979)} \\
        {\color[HTML]{005AFF} Men in Black (1997)} & {\color[HTML]{FF4B00} Star Wars: Episode IV (1977)} \\
        {\color[HTML]{005AFF} The Fugitive (1993)} & {\color[HTML]{FF4B00} Star Trek: The Motion Picture (1979)} \\ \bottomrule
    \end{tabular}
    }
    \label{tab: example}
\end{table}

\subsection{(RQ3) User Study} \label{sec: experiments-user}

We conducted user studies of \textsc{Consul} on Amazon Mechanical Turk.

\vspace{0.1in}
\noindent \textbf{Setup.} We used the MovieLens 1M dataset in this experiment. We chose a movie dataset because crowd workers can easily interpret the results without much technical knowledge \cite{serbos2017fairness}. We choose $100$ movies with the most reviews, i.e., $100$ most popular items, as source items. We use BPR for the provider's recommender system as in the previous experiment. A movie is in a protected group if its release date is more that $10$ years earlier or later than the source items. We set $K = 6$ and $\tau = 3$ in this experiment. We conduct two user studies. We show the crowd workers the two recommendation lists side by side as in Table \ref{tab: example}, and in the first study, we ask the workers which list is more relevant to the source item, and in the second study, we ask them which list is more divergent in terms of release dates. We show the list without any colors or method names, and thus, the workers do not know which method offers which list. To mitigate the position bias, we show the provider's recommender system on the left panel for $50$ trials, and \textsc{Consul} on the left panel for $50$ trials. The two experiments are run in different batches so that one result does not affect the other result.

\vspace{0.1in}
\noindent \textbf{Results.} Figure \ref{fig: user} reports the number of times each method was chosen. This shows that \textsc{Consul} is on par or slightly worse than the provider's recommender system in terms of relevance. Meanwhile, the crowd workers found \textsc{Consul} to be more divergent. This shows that \textsc{Consul} indeed improves fairness (here in terms of release dates) and recommends divergent items. We stress that \textsc{Consul} does not necessarily compensate the relevance to improve fairness. Table \ref{tab: example} shows an actual example where the crowd workers found \textsc{Consul} to be \emph{both} relevant and divergent. \textsc{Consul} succeeded in recommending \emph{classic} sci-fi movies, whereas the provider's recommended only \emph{contemporary} movies, which have indeed high scores from the perspective of collaborative filtering. This result shows that \textsc{Consul} can retrieve information that the provider recommender system cannot.

We note that the official and user-side recommender systems are not alternatives to each other. Rather, a user can consult both official and user-side recommender systems. The diversity of \textsc{Consul} is preferable even if it degrades the relevance to some extent because official and user-side recommender systems can complement the defects of one another.

\subsection{(RQ4) Case Study in the Real World} \label{sec: twitter}

\begin{table}[t]
    \caption{Recommendations for ``Hugh Jackman'' on the real-world Twitter environment. The provider's official system recommends only male users and is not fair with respect to gender. Both \textsc{PrivateWalk} and \textsc{Consul} succeed completely balanced recommendations. All the recommended users by \textsc{Consul} are actors/actresses and related to the source user, Hugh Jackman. Besides, the number of accesses to the official system is much less than in \textsc{PrivateWalk}.}
    \centering
    \vspace{-0.15in}
    \scalebox{0.8}{
    \begin{tabular}{ccc} \toprule
        \multicolumn{3}{c}{Hugh Jackman} \\ \midrule
        Provider & \textsc{PrivateWalk} & \textsc{Consul} \\
        N/A & 23 accesses & \textbf{5 accesses} \\ \midrule
        Chris Hemsworth ({\color[HTML]{FF4B00} man}) & Ian McKellen ({\color[HTML]{FF4B00} man}) & Chris Hemsworth ({\color[HTML]{FF4B00} man}) \\
        Chris Pratt ({\color[HTML]{FF4B00} man}) & Zac Efron ({\color[HTML]{FF4B00} man}) & Chris Pratt ({\color[HTML]{FF4B00} man})  \\
        Ian McKellen ({\color[HTML]{FF4B00} man}) & Seth Green ({\color[HTML]{FF4B00} man}) & Ian McKellen ({\color[HTML]{FF4B00} man}) \\
        Zac Efron ({\color[HTML]{FF4B00} man}) & Dana Bash ({\color[HTML]{005AFF} woman}) & Brie Larson ({\color[HTML]{005AFF} woman}) \\
        Patrick Stewart ({\color[HTML]{FF4B00} man}) & Lena Dunham ({\color[HTML]{005AFF} woman}) & Danai Gurira ({\color[HTML]{005AFF} woman}) \\
        Seth Rogen ({\color[HTML]{FF4B00} man}) & Jena Malone ({\color[HTML]{005AFF} woman}) & Kat Dennings ({\color[HTML]{005AFF} woman}) \\ \bottomrule
    \end{tabular}
    }
    \label{tab: twitter}
\end{table}

Finally, we show that \textsc{Consul} is applicable to real-world services via a case study. We use the user recommender system on Twitter, which is in operation in the real-world. We stress that we are not employees of Twitter and do not have any access to the hidden data stored in Twitter. The results below show that we can build a recommender system for Twitter even though we are not employees but ordinary users. The sensitive attribute is defined by gender in this experiment. Table \ref{tab: twitter} shows the results for the account of ``Hugh Jackman.'' While the official recommender system is not fair with respect to gender, \textsc{Consul}'s recommendations are completely balanced, and \textsc{Consul} consumes only $5$ queries, which is $4$ times more efficient than \textsc{PrivateWalk}. Note that ETP and \textsc{PrivateRank} require to download all users in Twitter, and thus they are infeasible for this setting. These results highlight the utility of \textsc{Consul} in real-world environments.

\section{Related Work}

\noindent \textbf{User-side Information Retrieval.} User-side (or client-side) information retrieval \cite{sato2022private, sato2022retrieving, sato2022clear} aims to help users of web services construct their own information retrieval and recommender systems, whereas traditional methods are designed for the developers of the web services. This setting has many unique challenges including limited resources and limited access to databases. Our study focuses on user-side recommender systems, a special case of user-side information retrieval problems. User-side information retrieval systems are closely related to focused crawling  \cite{chakrabarti1999focused, mccallum2000automating, johnson2003evolving, baezayates2005crawling, guan2008guide}, which aims to retrieve information of specific topics \cite{chakrabarti1999focused, mccallum2000automating}, popular pages \cite{baezayates2005crawling}, structured data \cite{meusel2014focused}, and hidden pages \cite{barbosa2007adaptive}. The main difference between focused crawling and user-side information retrieval is that focused crawling consumes a great deal of time, typically several hours to several weeks, whereas user-side information retrieval aims to conduct real-time inference. Our proposed method is highly efficient as shown in the experiments and enables inference even when costly crawling is prohibitive.

\vspace{0.1in}
\noindent \textbf{Fairness in Recommender Systems.} As fairness has become a major concern in society \cite{united2014big, executive2016big}, many fairness-aware machine learning algorithms have been proposed \cite{hardt2016equality, kamishima2012fairness, zafar2017fairness}. In particular, fairness with respect to gender \cite{zehlike2017fair, singh2018fairness, xu2020algorithmic}, race \cite{zehlike2017fair, xu2020algorithmic}, financial status \cite{fu2020fairness}, and popularity \cite{mehrotra2018towards, xiao2019beyond} is of great concern. In light of this, many fairness-aware recommendation algorithms have been proposed \cite{kamishima2012enhancement, yao2017beyond, biega2018equity, milano2020recommender, sato2022enumerating}. Fair recommender systems can be categorized into three groups \cite{bruke2017mltisided}. C-fairness concerns fairness for users; it, for example, ensures that female and male users should be treated equally. P-fairness concerns fairness for items; it, for example, ensures that news related to Republicans and Democrats are treated equally. CP-fairness is the combination of the two. In this study, we focus on P-fairness following \cite{sato2022private}. A notable application of P-fair recommender systems is in the job-market \cite{geyik2019fairness, geyik2018building}, where items correspond to job-seekers, and the sensitivity attributes correspond to gender and race. In contrast to traditional P-fair recommender systems \cite{ekstrand2018exploring, beutel2019fairness, mehrotra2018towards, liu2019personalized}, which are designed for the developers of the services, our proposed method is special in that it is designed for the users of the services. Note that fairness is closely related to topic diversification \cite{ziegler2005improving} by regarding the topic as the sensitive attribute, and we consider the diversity of recommended items in this study as well.

\vspace{0.1in}
\noindent \textbf{Steerable Recommender Systems.} The reliability of recommender systems has attracted a lot of attention \cite{tintarev2007survey, balog2019transparent}, and steerable recommender systems that let the users modify the behavior of the system have been proposed \cite{green2009generating, balog2019transparent}. User-side recommender systems also allow the users modify the recommendation results. However, the crucial difference between steerable and user-side recommender systems are that steerable recommender systems must be implemented by a service provider, whereas user-side recommender systems can be built by arbitrary users even if the official system is an ordinary (non-steerable) one. Therefore, user-side recommender systems can expand the scope of steerable recommender systems by a considerable margin \cite{sato2022private}.

\vspace{0.1in}
\noindent \textbf{Metric Recovery.} We utilize the results on metric recovery when we derive the feasibility of user-side recommender systems. Recovering hidden representations has a long history, including the well-known multi-dimensional scaling \cite{kruskal1964multidimensional, agarwal2007generalized} and Laplacian Eigenmaps \cite{belkin2003laplacian}. In particular, our discussions rely on the metric recovery theory on unweighted $k$-NN graphs \cite{alamgir2012shortest, hashimoto2015metric, luxburg2013density, terada2014local}. To the best of our knowledge, this work is the first to connect the theory to security, fairness, or user-side aspects of recommender systems.

\section{Conclusion}

We first saw that interested users can estimate hidden features of items solely from recommendation results. This theory indicates that the raw recommendation results contain sufficient information to build user-side recommender systems and elucidates the rigorous feasibility of building user-side recommender systems without using log data. However, this approach is not practical due to its high communication costs. To design practical user-side recommender systems, we proposed three desirable properties of user-side recommender systems: consistency, soundness, and locality. We found that existing user-side recommender systems lack at least one property. We, therefore, proposed \textsc{Consul}, a user-side recommender system that satisfies the three desired properties. We then demonstrated empirically that confidential features can indeed be recovered solely from recommendation results. We also confirmed that our proposed method is much more communication efficient than the existing methods while retaining a high degree of accuracy. Finally, we conducted user studies with crowd workers and confirmed that our proposed method provided diverse items compared with provider's official recommender systems.

\begin{acks}
This work was supported by JSPS KAKENHI GrantNumber 21J22490. 
\end{acks}

% \begin{acks}
% \end{acks}

% \bibliographystyle{ACM-Reference-Format}
% \bibliographystyle{abbrv}
\bibliographystyle{plainnat}
\bibliography{sample-base}

\end{document}